\documentclass[runningheads]{llncs}
\usepackage{graphicx}
\usepackage{xcolor}
\usepackage{amsmath}
\usepackage{amsfonts}
\usepackage{mismath}
\usepackage{hyperref}
\usepackage[framemethod=TikZ]{mdframed}
\usepackage{appendix}
\usepackage{booktabs}
\usepackage{dashrule}
\usepackage{multirow}
\usepackage{pifont}
\usepackage{ifthen}
\usepackage{xfp}
\usepackage{float}
\usepackage{enumitem}
\usepackage{orcidlink}
\newcommand{\lsb}[1]{%
   \ifthenelse{ \equal{#1}{} }
      {\ensuremath{\mathrm{LSB}}}
      {\ensuremath{\mathrm{LSB}(#1)}}
}
\newcommand{\bsl}[1]{%
   \ifthenelse{ \equal{#1}{} }
      {\ensuremath{\mathrm{L}}}
      {\ensuremath{\mathrm{L}(#1)}}
}
\newcommand{\garblescheme}{\ensuremath{\mathrm{GC}}}
\newcommand{\garble}{\ensuremath{\mathrm{GC.Gb}}}
\newcommand{\encode}{\ensuremath{\mathrm{GC.Enc}}}
\newcommand{\eval}{\ensuremath{\mathrm{GC.Eval}}}
\newcommand{\decode}{\ensuremath{\mathrm{GC.Dec}}}

\newcommand{\enc}{\ensuremath{\mathnormal{Enc}}}
\newcommand{\dec}{\ensuremath{\mathnormal{Dec}}}
\newcommand{\skippair}{\ensuremath{\mathrm{Pair}}}
\newcommand{\genpss}{\ensuremath{\mathrm{PSS.Gen}}}
\newcommand{\evalpss}{\ensuremath{\mathrm{PSS.Eval}}}
\newcommand{\gencss}{\ensuremath{\mathrm{CSS.Gen}}}
\newcommand{\evalcss}{\ensuremath{\mathrm{CSS.Eval}}}
\newcommand{\hash}{\ensuremath{\mathnormal{H}}}
\newcommand{\view}[1]{%
   \ifthenelse{ \equal{#1}{} }
      {\ensuremath{\mathrm{View}}}
      {\ensuremath{\mathrm{View}(#1)}}
}
\newcommand{\skippingscheme}{\ensuremath{\mathrm{SS}}}
\newcommand{\genss}{\ensuremath{\mathrm{SS.Gen}}}
\newcommand{\evalss}{\ensuremath{\mathrm{SS.Eval}}}

\newcommand{\bigtheta}{\ensuremath{\mathup{\Theta}}}
\newcommand{\game}{\ensuremath{\mathcal{G}}}
\newcommand{\topo}{\ensuremath{\Phi_{\text{topo}}}}
\newcommand{\ad}{\ensuremath{\mathcal{A}}}
\newcommand{\adv}{\ensuremath{Adv}}
\newcommand{\Sim}{\ensuremath{\mathrm{Sim}}}

\newcommand{\oracle}{\ensuremath{\mathcal{O}}}
\newcommand{\tidx}[1]{\fpeval{1 + #1}}

\newenvironment{schemealg}[1][]{%
    \ifstrempty{#1}{}%
    {
        \mdfsetup{%
            frametitle={%
                \tikz[baseline=(current bounding box.east),outer sep=0pt]
                \node[anchor=east,rectangle,draw=black,fill=white,line width=1pt,rounded corners=5pt]
                {\strut #1};}}%
    }
    \mdfsetup{innertopmargin=0pt,innerbottommargin=5pt,linecolor=black,%
        linewidth=1pt,topline=true,
        roundcorner=5pt,
        frametitleaboveskip=\dimexpr-\ht\strutbox\relax,}%
    \begin{mdframed}[]\relax%
}{\end{mdframed}}%

\newenvironment{normalalg}[1][]{%
    \ifstrempty{#1}{}
    {
        \mdfsetup{
            frametitle={\strut #1},
            frametitlerule=true,
        }
    }
    \mdfsetup{
        linewidth=1pt,
        innertopmargin=5pt,
        innerbottommargin=5pt,
        roundcorner=5pt,
    }
    \begin{mdframed}[]\relax
}{\end{mdframed}}

\newboolean{anony}
\setboolean{anony}{false}    % camera VERSION

\makeatletter
\def\blfootnote{\xdef\@thefnmark{}\@footnotetext}
\makeatother

\begin{document}
\title{Skipping Scheme for Gate-hiding Garbled Circuits}
\ifthenelse{\boolean{anony}}{
    \author{Anonymous}
    \authorrunning{Anonymous}
}{
    \author{Ke Lin\inst{1}\orcidlink{0009-0002-5376-7881}
    % \and Yasir Glani\inst{1}\orcidlink{0000-0003-0060-4771}
    % \and Dinghong Song\inst{1}\orcidlink{0009-0009-6986-963X}
    % \and Ping Luo\inst{1}\orcidlink{0000-0001-6171-3811}
    }
    \authorrunning{K. Lin et al.}
    \institute{
    % Tsinghua University, Peking, China\\
        \email{leonard.keilin@gmail.com}\\
        % \email{yasirglani@gmail.com}\\
        % \email{sdh21@mails.tsinghua.edu.cn}\\
        % \email{luop@mail.tsinghua.edu.cn}
    }
}
\maketitle              % typeset the header of the contribution
\begin{abstract}
%
%The abstract should briefly summarize the contents of the paper in
%15--250 words.
%
In classic settings of garbled circuits, each gate type is leaked to improve both space and speed optimization. Zahur et al. have shown in EUROCRYPT 2015 that a typical linear garbling scheme requires at least two $\lambda$-bit elements per gate with a security parameter of $\lambda$, which limits their efficiency. In contrast to typical garbled circuits, gate-hiding garbled circuits have the potential to drastically reduce time costs, although they have been underappreciated.

We propose the first skipping scheme for gate-hiding garbled circuits to enhance the efficiency of evaluation by observing prime implicants. Our scheme introduces skip gates to eliminate the need to calculate the entire circuit, enabling unnecessary execution paths to be avoided. We also introduce two variants of our scheme that balance security with parallelism. A proof of hybrid security that combines simulation-based and symmetry-based security in semi-honest scenarios is presented to demonstrate its security under gate-hiding conditions.
Our scheme will inspire new directions to improve the general garbling scheme and lead to more practical ones.


\end{abstract}

%
% Introduction
%
\section{Introduction}
% \ifthenelse{\boolean{anony}}{
% % pass
% }{
% \blfootnote{This work is supported by National Key R\&D Program of China under grant (No.2022YFB2703001).}
% }

Garbled circuit (GC) was first introduced by Andrew Yao \cite{yao1982protocols} that enables secure two-party computation. 
A circuit $C$ and an input $x$ are often used to derive a garbled circuit $\tilde{C}$ and a garbled input $X$, according to the typical pipeline of garbling schemes. 
The correct result $C(x)$ may be obtained by evaluating $\tilde{C}$ on $X$ and decrypting from $\tilde{C}(X)$ without disclosing anything more. 
In particular, a simulator can simulate $\tilde{C}$ and $X$ given only $C(x)$.

In a typical two-party computation scenario, the circuit $C$ that each party wishes to assess is chosen by consensus \cite{pinkas2009secure}. The garbled circuits in these protocols only need to conceal the inputs to $C$ because $C$ is a public variable and therefore does not need to be kept secret in any other way. 
It should be noted, however, that in other situations, such as private function evaluation (PFE), it can be advantageous for a garbled circuit to mask any identifying information about the circuit itself \cite{mohassel2013hide}.

Unlike classical garbled circuits, gate-hiding garbled circuits allow only the topology to be leaked. At the same time, the gate functions remain a secret. Although most optimizations on gate-aware circuits are based on the leakage of gate types, we demonstrate in this paper that keeping the gate function unknown has the potential to reduce the time costs of evaluating subcircuits quasi-exponentially, resulting in a practical garbling scheme.

\subsection{Related Works}
The garbled circuit technique described by Yao was promoted to a cryptographic primitive, the \emph{garbling scheme}, by Bellare et al. \cite{bellare2012foundations}, which is used along with our scheme in this work. 
Since this technique has been introduced, numerous efforts \cite{saleem2018recent} have been made to improve it. Techniques such as Point and Permute \cite{beaver1990round}, GRR3 \cite{naor1999privacy}, FreeXOR \cite{kolesnikov2008improvedfreexor} and Half Gates \cite{zahur2015twohalf} improve the efficiency of execution by significantly reducing the space requirements. Since Fairplay \cite{malkhi2004fairplay}, elegant implementations of garbled circuits have begun to fill up the gap between theory and practice, including FastGC \cite{huang2011fastgc}, Obliv-C \cite{zahur2015oblivc}, and TinyGarble \cite{songhori2015tinygarble}. Despite the advances, Zahur et al. \cite{zahur2015twohalf} argue that all linear, and thus, efficient, garbling schemes must contain at least two $\lambda$-bit elements to function properly. Even though some works attempt to reduce the lower bound of these gates to $\bigo(1.5\lambda)$ \cite{kempka2016circumvent,rosulek2021threehalf}, it appears that the efficiency of classic garbled circuits is reaching its limit.

In contrast to gate-aware garbled circuits, gate-hiding circuits remain underdeveloped and underappreciated, despite the fact that Yao's original protocol achieves gate-obliviousness. 
There are several garbling schemes designed to support both AND and XOR gates while masking their type from the evaluator \cite{kempka2016circumvent,rosulek2017improvements,wang2017reducing}. The varieties differ in the types of Boolean gates they can support.
Recently, Rosulek et al. \cite{rosulek2021threehalf} achieved the state-of-the-art gate-hiding garbling, in which \emph{any} gate is garbled for $1.5\lambda+10$ bits. Furthermore, their gate-hiding construction is fully compatible with FreeXOR.

Songhori et al. \cite{songhori2019arm2gc} proposed a scheme in 2019 that is very similar to ours to improve the efficiency of evaluating garbled circuits using SkipGate in the ARM2GC framework. However, their method only allows efficient secure evaluation of functions with publicly known inputs. 
In addition, Kolesnikov \cite{kolesnikov2018free} proposes using FreeIF to omit inactive branches from garbled circuits for free, which is an analogous idea to ours but with online communication.
% We extend skip gates to general gate-hiding cases, evaluate their efficiency theoretically, and provide proof that hybrid security can be achieved.

\subsection{Our Contributions}

We propose an effective skipping scheme at runtime for gate-hiding garbled circuits. The scheme skips inaccessible execution pathways and promotes parallelism on the fly. 
In addition to significantly improving the evaluation speed, our approach is simple to understand and use. 
The motivation for this approach stems from the observation that prime implicants in Boolean algebra disregard insignificant inputs when evaluating functions.
% Using our method, semi-private functions (SPF-SFE) can be evaluated securely when the evaluator is provided only with the circuit topology and not with the gate functions.
The method can also be used when the circuit consists of both plain and hidden gates, in which case the method is applied only to hidden gates.
We prove that our skipping scheme achieves hybrid privacy.

A plain skipping scheme (PSS) and a chained skipping scheme (CSS) are available in our construction. 
In both schemes, skip gates are constructed and can be activated or \emph{triggered} based on certain requirements.
% Some "skip gates" are introduced at the beginning of subcircuits as part of the skipping scheme.
% In PSS, each skip gate calculates its final result based on two ciphertexts of $\lambda$ bits and two trigger values that indicate whether the skipping scheme is active. 
% Each CSS gate requires two additional chain values to provide enhanced security protection.
% It is pertinent to point out that the extra information provided by CSS effectively prevents leakage of the ground-truth output value of the gates that are not triggered.
%
We demonstrate how unnecessary execution paths can be avoided by utilizing the skipping scheme in semi-private garbled circuits. 
Even though there are more gates to generate, the reducing in executions saves evaluation time and computing resources. 
In addition to validating the skipping scheme's correctness, a proof of hybrid security is presented to prove its security under gate-hiding conditions. 
% Furthermore, it indicates that SPF-SFE can enhance parallel computing so that it can be applied in more practical scenarios.

The efficiency of (sub)circuit evaluation is generally improved by our scheme, which eliminates the need to calculate the entire circuit at once. We also obtain a significant improvement for subcircuits since we only need to use a fraction of the inputs. It optimizes the calculation of gates that remain unused and benefit from parallel computation.

\section{Preliminaries}

\subsection{Notation}

The notations are defined as follows: $x\xleftarrow{U} X$ represents that $x$ was sampled randomly from $X$ under a uniform distribution, and $x\gets\mathrm{Alg}$ indicates that $x$ is generated using the algorithm $\mathrm{Alg}$. The symbol $\oplus$ signifies the exclusive OR operation on bitstrings, and $[n]$ is an alias for $\{1,\dots,n\}$. $\lsb{S}$ denotes the least significant bit of a bitstring $S$, while the remainder is represented by $\bsl{S}$. The security parameter is $\lambda$.

\subsection{Garbled Circuits}
\label{sec:garbled_circuits}

This section presents a formal definition of a garbling scheme and the notion of simulation-based privacy, following the work of Bellare et al. \cite{bellare2012foundations}. Note that all the gates in this work accept two inputs.

We use notation from Kempka et al. \cite{kempka2016circumvent}, which describes a circuit as 
$$C:=(n,m,l,A,B,g)$$
where $n$ is the number of input wires, $m$ is the number of output wires, and $l$ is the number of gates. 
A set of wires can be considered as $W:=[n+l]$, a circuit input wire set as $W_{\text{input}}: = [n+l]\setminus [n+l-m]$, and the circuit output wires as $W_{\text{gate}}:=[n+l]\setminus [n]$. 
The functions $A: W_{\text{gate}}\to W\setminus W_{\text{output}}$ and $B: W_{\text{gate}}\to W\setminus W_{\text{output}}$ indicate the first input wire of each gate $i$ and the second input wire $B(i)$. 
Each gate $i$ is specified by the function $g(i,\cdot,\cdot)=g_i(\cdot,\cdot)$.
For simplicity, we will omit the parameter $i$ in some cases.

\begin{definition}[Garbling Scheme]
A garbling scheme for a family of circuits $\mathcal{C}=\{\mathcal{C}_n\}_{n\in\mathbb{N}}$, where $n$ is a polynomial in a security parameter $k$, employs probabilistic polynomial-time algorithms $\garblescheme=(\garble,\encode,\eval,\decode)$ defined as follows:
\begin{itemize}
    \item $(\tilde{C},e,d)\gets\garble(1^\lambda,C)$: takes as input security parameter $1^\lambda$ and circuit $C\in \mathcal{C}$ and outputs garbled circuit $\tilde{C}$, encoding information $e$ and decoding information $d$.
    \item $X\gets\encode(e,x)$: takes as input encoding information $e$ and circuit input $x\in \{0,1\}^n$, and outputs garbled input $X$.
    \item $Y\gets\eval(\tilde{C},X)$: takes as input garbled circuit $\tilde{C}$ and garbled input $X$, and outputs garbled output $Y$.
    \item $y\gets\decode(d,Y)$: takes as input decoding information $d$ and garbled output $Y$, and outputs circuit output $y$.
\end{itemize}

A garbling scheme should satisfy the \textnormal{correctness} property: For any circuit $C$ and input $x$, after sampling $(\tilde{C},e,d)$ from the garbling scheme, and the following equation $C(x)=\decode(d,\eval(\tilde{C},\encode(e,x)))$ holds with all but negligible probability.
\end{definition}

% We provide a definition of simulation-based privacy for garbling schemes. 
We introduce the notion of a random oracle $H$ to allow adversary access. We also indicate the information on circuit $C$ that is permitted to be leaked by the garbling scheme by the value $\Phi(C)$. Topology $\Phi_{\text{topo}}(C)=(n,m,l,A,B)$ or entire circuit information $\Phi_{\text{circ}}(C)=(n,m,l,A,B,g)$, for example, can be leaked. And the gate-hiding garbling scheme always sends $\Phi_\text{topo}$ instead of the entire circuit to the evaluator.

% \begin{definition}[Simulation-based Privacy]
% For a garbling scheme $\garblescheme=(\garble,\encode,\eval,\decode)$, circuit $C\in\mathcal{C}$, input value $x\in\{0,1\}^n$, simulator $\mathrm{Sim}$ and a random oracle $H$, the advantage of the adversary $\ad$ is defined as:
% \begin{equation*}
% \begin{aligned}
% \Big\lvert&\Pr[(\tilde{C},X,d)\gets\mathrm{Sim}(1^k,\Phi(C),C(x)):\ad^H(\tilde{C},X,d)=1] \\
% -&\Pr[(\tilde{C},e,d)\gets\garble(1^k,C),X\gets\encode(e,x):\ad^H(\tilde{C},X,d)=1]\Big\rvert
% \end{aligned}
% \end{equation*}
% A garbling scheme $\mathrm{GC}=(\garble,\encode,\eval,\decode)$ is \textnormal{private}, if there exists a PPT simulator $\mathrm{Sim}$, such that for any circuit $C\in \mathcal{C}$, input values $x\in \{0,1\}^n$, and PPT adversary $\ad$, the advantage is negligible.
% \end{definition}

%
% Kinda definition of skipping scheme
%
\section{Skipping Scheme}
\label{sec:skipping_scheme}

This section presents our basic skipping scheme for semi-honest gate-hiding circuits. The scheme is first demonstrated for circuits with 2-implicative components. Then, it is extended from the initial scenario for circuits with 2-implicative components to a more general scenario for $n$-implicative components. 
% In Section \ref{sec:construction}, we present the construction of two types of skipping schemes, along with the extension for $n$-implicative circuits. Section \ref{sec:general_pipeline} provides an overview of the general pipeline of the garbling scheme along with the skipping scheme. Section \ref{sec:efficiency} shows the estimated efficiency of the skipping scheme with a considerable reduction in computational overhead.

\subsection{Construction}
\label{sec:construction}

We use the following notation. Let $\alpha:=\lito(\lambda)\ll\lambda$ be the security parameter for hash functions, which can usually be a constant. The input wires $W_i$ may be assigned a value of 0 or 1, which is indicated by $W_i^0,W_i^1$. The set of possible assignments of two input wires $W_i,W_j$ is denoted by $\mathcal{V}_{i,j}:=\{(W_i^{b_i},W_j^{b_j})\mid b_i,b_j\in\{0,1\}\}$. We then define the $k$-implicative (sub)circuit and skippable wires formally.

\begin{definition}[$k$-implicative]
\label{def:k-implicative}
Let $C$ be a Boolean circuit with $n$ input wires $\mathcal{W}=\{W_i\}^n_{i=1}$. We say circuit $C$ is \textnormal{$k$-implicative} ($k<n$) if: there exist $k$ fixed assignments $\{W^{b_{d_i}}_{d_i}\}^k_{i=1}$ to a subset of input wires $\mathcal{K} = \{W_{d_i}\}^{k}_{i=1}\subset\mathcal{W}$ such that for every assigments $\{W^{x_{\overline{d_i}}}_{\overline{d_i}}\}^{n-k}_{i=1}$ to $\mathcal{W}\backslash\mathcal{K}$,
$$
C(W^{x_1}_1,\dots,W^{b_{d_1}}_{d_1},\dots,W^{b_{d_k}}_{d_k},\dots,W^{x_n}_n)\equiv \text{constant}
$$
\end{definition}
\begin{corollary}
A $n$-implicative circuit is also $m$-implicative for all $m>n$.
\end{corollary}
% \begin{definition}[$k$-implicative]
% \label{def:k-implicative}
% A Boolean circuit $C$ with $n$ inputs is $k$-implicative ($k\le n$) if for its ground-truth function $F_C$, there exist $k$ fixed inputs $\mathcal{K} = \{I_{d_i}\}_{k}^{i=1}\subseteq \mathcal{I}$ and a constant $c$ such that
% $$
% \forall I_i \notin \mathcal{K}, F_C(I_1,\dots,I_n) = c
% $$
% \end{definition}
% \begin{remark}
% In Boolean algebra, an \emph{implicant} is another form of $k$-implicative wire. In particular use, a product term (i.e., a conjunction of literals) $P$ implicates a Boolean function $F$, if $P$ implies $F$ (i.e., whenever $P$ takes the value 1, so does $F$).
% A $k$-implicative circuit is equivalent to the size-$k$ implicant in the area of Boolean logic. However, the former is used for garbled circuits in this work.
% \end{remark}
A $k$-implicative circuit implies that $k$ specific assignments to parts of the inputs are sufficient to compute the final result. This simple observation leads to the reduction of evaluation costs by skipping the gates which are not the implicants.

\begin{definition}[Skippable]
\label{def:skippable}
For a circuit $C$, two of its input wires $W_i$ and $W_j$ are \textnormal{skippable} if there exist 2 different assignment pairs $v_0^{i,j},v_1^{i,j}\in\mathcal{V}_{i,j}$ that make the circuit $C$ 2-implicative in both cases. These two input wires are also known as a pair of skippable wires, denoted by $(W_i,W_j)\gets \skippair(C)$.
\end{definition}
% The remaining two pairs of labels are not needed since they are not included in the calculation to maintain the probability of triggering at $\frac{1}{2}$. Thus, it is impossible to distinguish the type of gate from another. 

Skippable wires are the basic elements of our skipping scheme. Definition \ref{def:k-implicative} discusses $n$-implicative wires as another form of implicants in Boolean algebra.
Previous research has shown that prime implicants of functions with many inputs can be found using heuristic methods, which is essentially more efficient by several orders of magnitude. The entire algorithm for searching for skippable wires can be found in \ref{appendix:skippable_wires_searching}. Next, we describe our scheme, which leverages these skippable wires. It should be noted that the skipping scheme targets the subcircuits of the entire circuit, and we will be abusing the terms subcircuit and circuit in the following sections.

\subsubsection{Plain Skipping Scheme.} Let $C$ be a Boolean circuit and $\tilde{C}\gets\garble(1^\lambda,C)$ be its garbled version. Let $\Gamma=(G,\enc,\dec)$ be a CPA-secure symmetric key encryption scheme, and $\hash_\alpha$ be a hash function with security parameter $\alpha$. 
The complete PSS is described in Fig. \ref{fig:pss}.

Trigger value $t$ is an easy-to-compute indicator that provides a fast method to determine whether the skip gates should be activated. 
The constant result of (sub)circuit is precomputed and encrypted as $G_b$ for $b\in\{0,1\}$. As there are two pairs of inputs $v_0,v_1$ that lead to the constant output, ciphertexts are permuted randomly to ensure security. Permute bit $b$ indicates which ciphertext should be used when the skip gate is triggered.
% To simplify the construction, we first describe the procedure for PSS informally, denoted by $\genpss$ and $\evalpss$ respectively. Considering the clear context, we may disregard the superscript for indicating skippable pairs. When $(W_i,W_j)\gets\skippair(C)$ produces the assignment pairs $v_0,v_1\in\mathcal{V}_{i,j}$, the procedure of generating PSS skip gates $\gp(1^\lambda,1^\alpha,\tilde{C},(W_i,W_j))$ produces 2 trigger values $(t_0,t_1)$ and 2 ciphertexts $(G_0,G_1)$:
% \begin{equation*}
% t_0\gets \hash_\alpha(s,v_0),~t_1\gets \hash_\alpha(s,v_1)
% \end{equation*}
Here $s$ is a shared salt to restrict the permute bit $b_0:=\lsb{t_0},b_1:=\lsb{t_1}$ and the remainder bitstring, such that $b_0\ne b_1$ and $\bsl{t_0}\ne \bsl{t_1}$. Then we compute the ciphertexts
\begin{equation*}
G_b\gets \enc_{v_0}^\lambda(R),~G_{1-b}\gets \enc_{v_1}^\lambda(R)
\end{equation*}
where $b:=b_0$ and $R:=\tilde{C}(\dots,v_0,\dots)=\tilde{C}(\dots,v_1,\dots)$ is the result label of the garbled circuit given the input $v_0$ or $v_1$ according to the Definitions \ref{def:k-implicative} and \ref{def:skippable}. Note that some superscripts, such as $i,j$, may be ignored if the reference to the wire pairs is clear.
% To prevent the permute bit $b$ from being acquired by the evaluator, the order of ciphertexts sent to the evaluator is fixed to $(G_0,G_1)$.

For simplicity, we denote the skip gates by $S_{i,j}^\text{PSS}:=(\bsl{t_0^{i,j}},\bsl{t_1^{i,j}},G_0^{i,j},G_1^{i,j})$. 

\begin{figure}[htb]
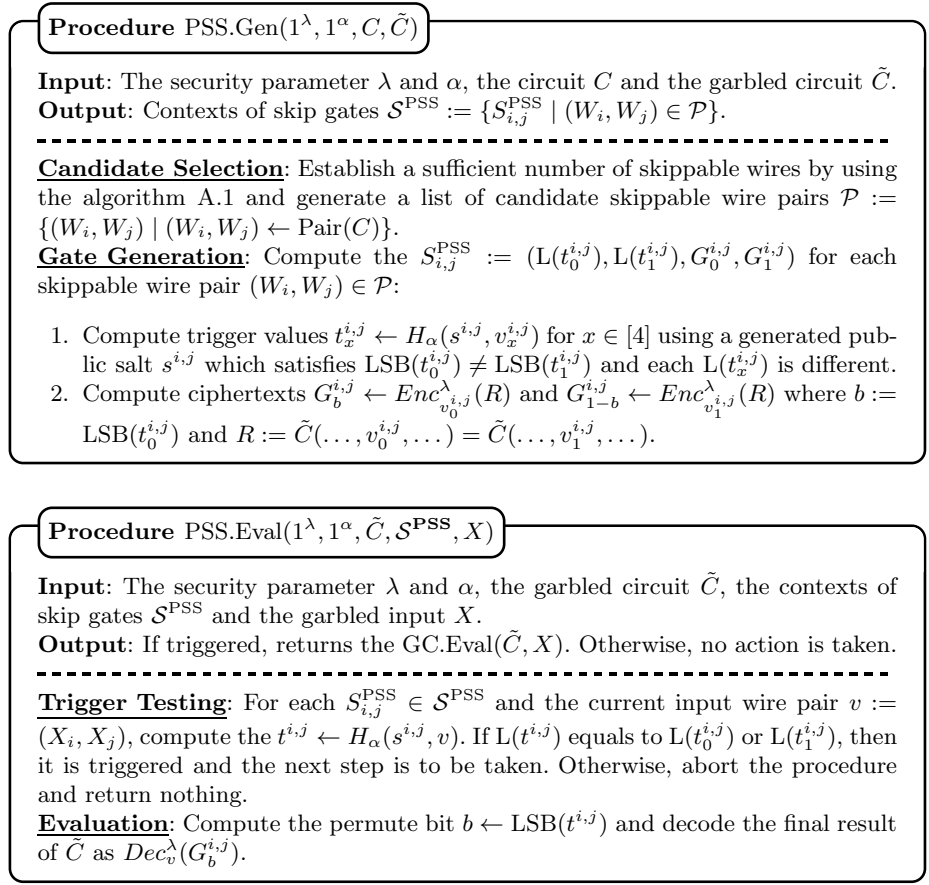

\centering
% GenPSS
\begin{schemealg}[\textbf{Procedure} $\genpss(1^\lambda,1^\alpha,C,\tilde{C})$]
\textbf{Input}: The security parameter $\lambda$ and $\alpha$, the circuit $C$ and the garbled circuit $\tilde{C}$.

\textbf{Output}: Contexts of skip gates $\mathcal{S}^\text{PSS}:=\{S_{i,j}^\text{PSS}\mid (W_i,W_j)\in \mathcal{P}\}$.
\newline
\hdashrule[0.5ex]{\linewidth}{1pt}{1mm}

\underline{\textbf{Candidate Selection}}: Establish a sufficient number of skippable wires by using the algorithm \ref{appendix:skippable_wires_searching} and generate a list of candidate skippable wire pairs $\mathcal{P}:=\{(W_i,W_j)\mid (W_i,W_j)\gets\skippair(C)\}$.

\underline{\textbf{Gate Generation}}: Compute the $S_{i,j}^\text{PSS}:=(\bsl{t_0^{i,j}},\bsl{t_1^{i,j}},G_0^{i,j},G_1^{i,j})$ for each skippable wire pair $(W_i,W_j)\in\mathcal{P}$:
\begin{enumerate}
    \item Compute trigger values $t_x^{i,j}\gets \hash_\alpha(s^{i,j},v_x^{i,j})$ for $x\in[4]$ using a generated public salt $s^{i,j}$ which satisfies $\lsb{t_0^{i,j}}\ne \lsb{t_1^{i,j}}$ and each $\bsl{t_x^{i,j}}$ is different.
    \item Compute ciphertexts $G_b^{i,j}\gets \enc_{v_0^{i,j}}^\lambda(R)$ and $G_{1-b}^{i,j}\gets \enc_{v_1^{i,j}}^\lambda(R)$ where $b:=\lsb{t_0^{i,j}}$ and $R:=\tilde{C}(\dots,v_0^{i,j},\dots)=\tilde{C}(\dots,v_1^{i,j},\dots)$.
\end{enumerate}
\end{schemealg}
% EvalPSS
\begin{schemealg}[\textbf{Procedure} $\evalpss(1^\lambda,1^\alpha,\tilde{C},\mathcal{S}^\text{PSS},X)$]
\textbf{Input}: The security parameter $\lambda$ and $\alpha$, the garbled circuit $\tilde{C}$, the contexts of skip gates $\mathcal{S}^\text{PSS}$ and the garbled input $X$.

\textbf{Output}: If triggered, returns the $\eval(\tilde{C},X)$. Otherwise, no action is taken.
\newline
\hdashrule[0.5ex]{\linewidth}{1pt}{1mm}

\underline{\textbf{Trigger Testing}}: For each $S_{i,j}^\text{PSS}\in\mathcal{S}^\text{PSS}$ and the current input wire pair $v:=(X_i,X_j)$, compute the $t^{i,j}\gets \hash_\alpha(s^{i,j},v)$. If $\bsl{t^{i,j}}$ equals to $\bsl{t_0^{i,j}}$ or $\bsl{t_1^{i,j}}$, then it is triggered and the next step is to be taken. Otherwise, abort the procedure and return nothing.

\underline{\textbf{Evaluation}}: Compute the permute bit $b\gets \lsb{t^{i,j}}$ and decode the final result of $\tilde{C}$ as $\dec_v^\lambda(G_b^{i,j})$.
\end{schemealg}
\caption{Procedure of the plain skipping scheme.}
\label{fig:pss}
\end{figure}

\subsubsection{Chained Skipping Scheme.} In contrast to PSS, CSS could form a decryption chain when evaluating the skip gates. A CSS chain can only be triggered if the previous gates in the chain have failed; this prevents leakage of the trigger state of subsequent gates.

In a similar fashion to PSS, we reuse many notations. Let $n$ be the number of skip gates in the chain. 
The garbler precomputes the context values for the $k$-th gate in the chain, which bring the previous skip gate's information to the current skip gate.
\begin{equation}
\label{eqa:css_context}
\begin{aligned}
c_\alpha^k\gets \hash_\alpha(s^k,c_\alpha^{k-1},T_\alpha(\hash_\alpha(s^k,c_\alpha^{k-1},v_3),\hash_\alpha(s^k,c_\alpha^{k-1},v_4))) \\
c_\lambda^k\gets \hash_\lambda(s^k,c_\lambda^{k-1},T_\lambda(\hash_\lambda(s^k,c_\lambda^{k-1},v_3),\hash_\lambda(s^k,c_\lambda^{k-1},v_4)))
\end{aligned}
\end{equation}
where $s^k$ is the salt for the $k$-th gate and $T_x:\{0,1\}^2\to \{0,1\}$ is an irreversible function that is abelian with respect to its two inputs for $x\in\{\alpha,\lambda\}$. 
Assignment pairs $v_3$ and $v_4$ are the pairs that are not included in the definition of 2-implicative, as opposed to their counterparts $v_0$ and $v_1$.
% In interpreting the definition, it is critical to bear in mind that $\tilde{C}(\dots,v_3,\dots)$ and $\tilde{C}(\dots,v_4,\dots)$ are not constrained in any way. This implies that they may even be the same as $R$ discussed previously.

% The context values bring the previous skip gate's information to the current skip gate. By introducing context values, the CSS skip gates' generating algorithm $\gc(1^\lambda,1^\alpha,\tilde{C},c_\lambda^{k-1},c_\alpha^{k-1},k)$ differs slightly from PSS:
% \begin{equation*}
% \begin{aligned}
% t_0^k\gets \hash_\alpha(s^k,c_\alpha^{k-1},v_0),&~t_1^k\gets \hash_\alpha(s^k,c_\alpha^{k-1},v_1) \\
% G_b^k\gets \enc_{v_0}^\lambda(R)\oplus c_\lambda^{k-1},&~G_{1-b}^k\gets \enc_{v_1}^\lambda(R)\oplus c_\lambda^{k-1}
% \end{aligned}
% \end{equation*}
% Specifically, both $c_\alpha^0$ and $c_\lambda^0$ could be defined as a constant, for example, a bitstring of zeros.

The generating algorithm for CSS skip gates produces two more chain values, $t_c^k$ and $G_c^k$.
\begin{equation}
\label{eqa:css_chain}
\begin{aligned}
t_c^k\gets \hash_\alpha(s^k,c_\alpha^{k-1},v_3)\oplus \hash_\alpha(s^k,c_\alpha^{k-1},v_4) \\
G_c^k\gets \hash_\lambda(s^k,c_\alpha^{k-1},v_3)\oplus \hash_\lambda(s^k,c_\alpha^{k-1},v_4)
\end{aligned} 
\end{equation}
These chain values are used to convey the other hash values without revealing them explicitly.
If $\hash_x(s^k,c_x^{k-1},v_y)$ for $x\in\{\alpha,\lambda\},y\in\{3,4\}$ is known when the previous gates have not been triggered, the chain values can then be used to calculate the hash $\hash_x(s^k,c_x^{k-1},v_{7-y})$. 

% Chain order must be followed when evaluating skip gates. As $t^k$ and $G^k$ are highly dependent on $c_\alpha^{k-1}$ and $c_\lambda^{k-1}$, any attempt to break this order would cause irreversible damage to the validity of the results.

The full CSS is shown in Fig. \ref{fig:css}. 
% While $\gencss$ generates contexts for skip gates, $\evalcss$ shows the process of evaluating them. 
% The most significant difference in $\gencss$ is the newly added step for context preprocessing, which calculates the context values before the gates are constructed. 
The pipeline of $\evalcss$'s evaluation is divided into two distinct branches. The first branch updates the context value, while the second is for evaluating the actual results. 
% We have differentiated the CSS from the PSS scheme by separating the differences from it while keeping the rest of the body more or less the same.

\begin{figure}[htb]
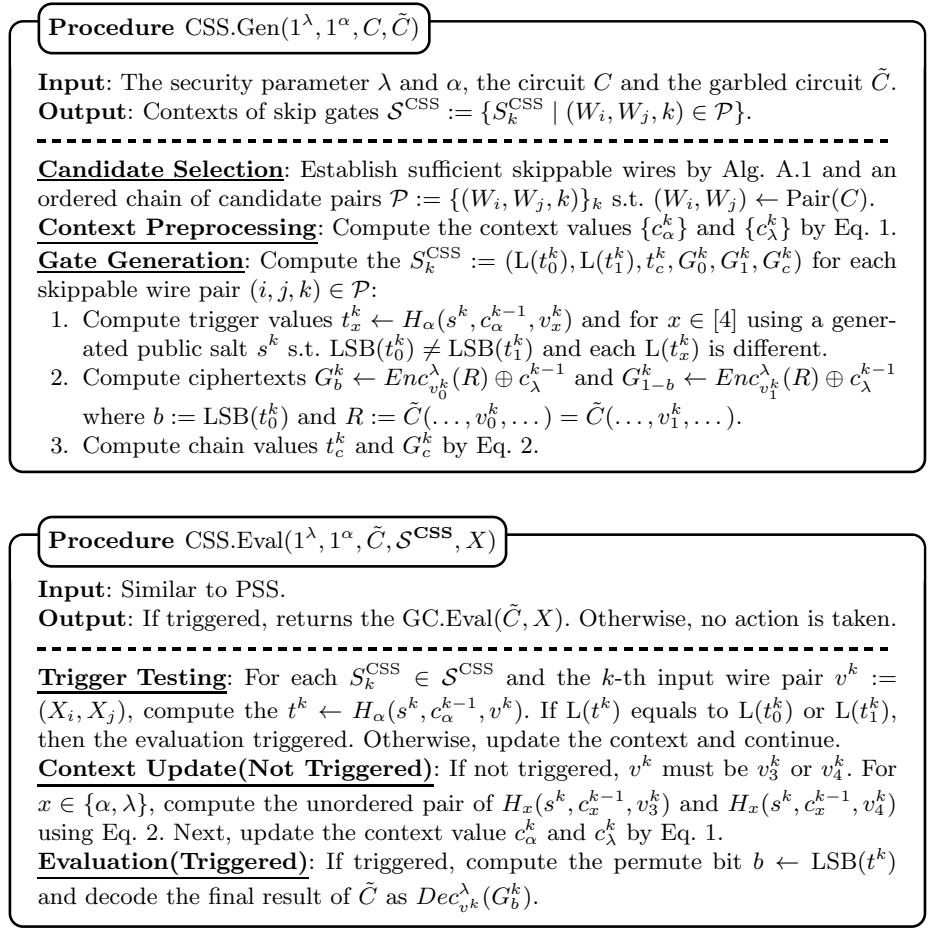

\centering
% GenCSS
\begin{schemealg}[\textbf{Procedure} $\gencss(1^\lambda,1^\alpha,C,\tilde{C})$]
\textbf{Input}: The security parameter $\lambda$ and $\alpha$, the circuit $C$ and the garbled circuit $\tilde{C}$.

\textbf{Output}: Contexts of skip gates $\mathcal{S}^\text{CSS}:=\{S_k^\text{CSS}\mid (W_i,W_j,k)\in \mathcal{P}\}$.
\newline
\hdashrule[0.5ex]{\linewidth}{1pt}{1mm}

\underline{\textbf{Candidate Selection}}: Establish sufficient skippable wires by Alg. \ref{appendix:skippable_wires_searching} and an ordered chain of candidate pairs $\mathcal{P}:=\{(W_i,W_j,k)\}_k$ s.t. $(W_i,W_j)\gets\skippair(C)$. 

\underline{\textbf{Context Preprocessing}}: Compute the context values $\{c_\alpha^k\}$ and $\{c_\lambda^k\}$ by Eq. \ref{eqa:css_context}.

\underline{\textbf{Gate Generation}}: Compute the $S_k^\text{CSS}:=(\bsl{t_0^k},\bsl{t_1^k},t_c^k,G_0^k,G_1^k,G_c^k)$ for each skippable wire pair $(i,j,k)\in\mathcal{P}$:
\begin{enumerate}[topsep=0pt,itemsep=-1ex,partopsep=1ex,parsep=1ex]
    \item Compute trigger values $t_x^k\gets \hash_\alpha(s^k,c_\alpha^{k-1},v_x^k)$ and for $x\in[4]$ using a generated public salt $s^k$ s.t. $\lsb{t_0^k}\ne \lsb{t_1^k}$ and each $\bsl{t_x^k}$ is different.
    \item Compute ciphertexts $G_b^k\gets \enc_{v_0^k}^\lambda(R)\oplus c_\lambda^{k-1}$ and $G_{1-b}^k\gets \enc_{v_1^k}^\lambda(R)\oplus c_\lambda^{k-1}$ where $b:=\lsb{t_0^k}$ and $R:=\tilde{C}(\dots,v_0^k,\dots)=\tilde{C}(\dots,v_1^k,\dots)$.
    \item Compute chain values $t_c^k$ and $G_c^k$ by Eq. \ref{eqa:css_chain}.
\end{enumerate}
\end{schemealg}
% EvalCSS
\begin{schemealg}[\textbf{Procedure} $\evalcss(1^\lambda,1^\alpha,\tilde{C},\mathcal{S}^\text{CSS},X)$]
\textbf{Input}: Similar to PSS.

\textbf{Output}: If triggered, returns the $\eval(\tilde{C},X)$. Otherwise, no action is taken.
\newline
\hdashrule[0.5ex]{\linewidth}{1pt}{1mm}

\underline{\textbf{Trigger Testing}}: For each $S_k^\text{CSS}\in\mathcal{S}^\text{CSS}$ and the $k$-th input wire pair $v^k:=(X_i,X_j)$, compute the $t^k\gets \hash_\alpha(s^k,c_\alpha^{k-1},v^k)$. If $\bsl{t^k}$ equals to $\bsl{t_0^k}$ or $\bsl{t_1^k}$, then the evaluation triggered. Otherwise, update the context and continue.

\underline{\textbf{Context Update(Not Triggered)}}: If not triggered, $v^k$ must be $v_3^k$ or $v_4^k$. For $x\in\{\alpha,\lambda\}$, compute the unordered pair of $\hash_x(s^k,c_x^{k-1},v_3^k)$ and $\hash_x(s^k,c_x^{k-1},v_4^k)$ using Eq. \ref{eqa:css_chain}. Next, update the context value $c_\alpha^k$ and $c_\lambda^k$ by Eq. \ref{eqa:css_context}.

\underline{\textbf{Evaluation(Triggered)}}: If triggered, compute the permute bit $b\gets \lsb{t^k}$ and decode the final result of $\tilde{C}$ as $\dec_{v^k}^\lambda(G_b^k)$.
\end{schemealg}
\caption{Procedure of the chained skipping scheme.}
\label{fig:css}
\end{figure}

\subsubsection{Extension for $n$-implicative Circuits.}
Let $\mathcal{W}_n:=\{W_i^{b_i}\}_{i=1}^{n}$ be the set of input wires that make the circuit $n$-implicative. Extensions to $n$-implicative circuits are based on the following techniques:
\begin{enumerate}
    \item Group the elements of the set $\mathcal{W}_n$ two by two to obtain the set of input wire pairs $\mathcal{V}_n:=\{(W_i^{b_i}, W_j^{b_j})\mid W_i^{b_i}, W_j^{b_j}\in \mathcal{W}_n\}$. Note that pairs can overlap to ensure that all input wires are covered.
    \item Establish an ordered chain of candidate input wire pairs that will never trigger in $\gencss$, except for the last.
    \item Substitute the above-generated candidate chain for the one used in gate generation for $\gencss$.
\end{enumerate}
% This method extends the skipping scheme for $n$-implicative circuits at the expense of generating too many unused trigger values and ciphertexts. Consequently, it is only feasible when the $n$-implicative circuit has a high probability of being triggered, as discussed in Section \ref{sec:efficiency}.

\subsection{General Pipeline}
\label{sec:general_pipeline}

Generally, the skipping scheme is constructed on top of the common garbling scheme. It works in parallel with the classic garbling circuit of \cite{bellare2012foundations} and leaves the garbling scheme intact. Our modified garbling scheme is shown in Fig. \ref{fig:patched_garbling_scheme}.

\begin{figure}[htb]
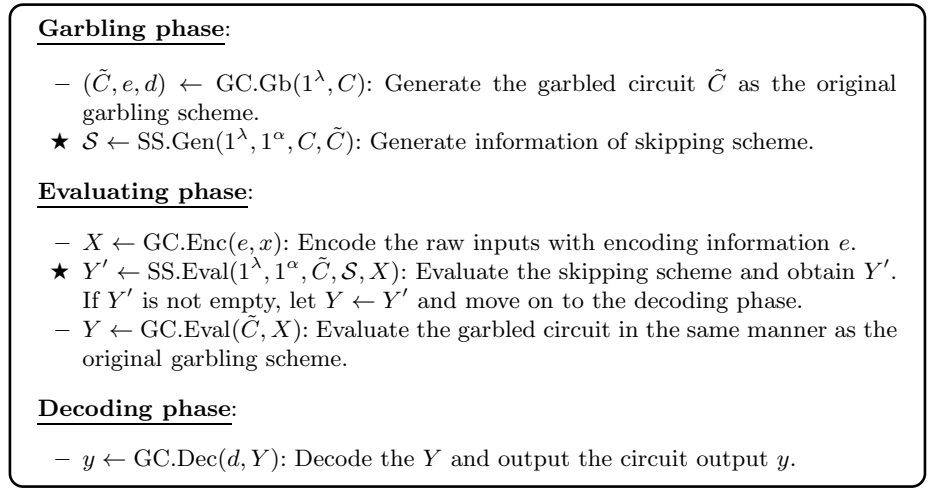

\centering
\begin{normalalg}
\underline{\textbf{Garbling phase}}:
\begin{itemize}
    \item $(\tilde{C},e,d)\gets\garble(1^\lambda,C)$: Generate the garbled circuit $\tilde{C}$ as the original garbling scheme.
    \item[\ding{72}] $\mathcal{S}\gets \genss(1^\lambda,1^\alpha,C,\tilde{C})$: Generate information of skipping scheme.
\end{itemize}

\underline{\textbf{Evaluating phase}}:
\begin{itemize}
    \item $X\gets\encode(e,x)$: Encode the raw inputs with encoding information $e$.
    \item[\ding{72}] $Y'\gets \evalss(1^\lambda,1^\alpha,\tilde{C},\mathcal{S},X)$: Evaluate the skipping scheme and obtain $Y'$. If $Y'$ is not empty, let $Y\gets Y'$ and move on to the decoding phase.
    \item $Y\gets\eval(\tilde{C},X)$: Evaluate the garbled circuit in the same manner as the original garbling scheme.
\end{itemize}

\underline{\textbf{Decoding phase}}:
\begin{itemize}
    \item $y\gets\decode(d,Y)$: Decode the $Y$ and output the circuit output $y$.
\end{itemize}
\end{normalalg}
\caption{Patched garbling scheme with skipping scheme. The algorithm $\mathrm{SS}$ here can be either $\mathrm{PSS}$ or $\mathrm{CSS}$, depending on the exact skipping scheme.}
\label{fig:patched_garbling_scheme}
\end{figure}
% The skipping scheme is integrated into the overall garbling scheme and does no effect on the execution of the garbled circuit. Therefore, it is still possible to use other optimization techniques for garbled circuits compatible with the gate-hiding cases.

\subsection{Efficiency}
\label{sec:efficiency}

Our efficiency estimation considers the skipping scheme's space complexity and computational complexity.
% In this study, the space complexity of the proposed gate construction scheme is compared with other existing practical gate construction schemes. Parallelism and computational complexity are assessed by estimating the cost of computational resources.

For a typical garbled circuit, we assume that each basic gate consume $2\lambda$ space \cite{zahur2015twohalf} as the starting points. PSS gates require $2\alpha+2\lambda$ space to store two trigger values and two ciphertexts. CSS gates require $3\alpha+3\lambda$ space to store two more chain values. We ignore the salts used in gate generation, as they consume very few bits in almost every case.

An estimation of the computational costs of each gate can be achieved by denoting the time costs of the hash function and the decryption by $t_h$ and $t_d$. We assume that the decryption time of the basic gate is the same as the skip gates, i.e. $t_d$ is used for decryption of both garbled gates and skip gates. 
Whenever a PSS gate is triggered, the evaluation cost is comprised of one hash for testing the PSS and one decryption for obtaining final results. Otherwise, if the PSS gate is not triggered, a hash calculation and a regular gate decryption is performed.
A similar analysis can be performed for CSS costs.

In this work, we estimate the complexity of a circuit consisting of $N$ gates. Let $n_c$ be the number of skip gates.
As each gate is triggered probabilistically, the computational complexity of each gate is analyzed by estimating both its expectation and the worst-case scenario separately.
We denote $\theta$ as the probability of triggering skip gates for general analysis. The uniform case without prior knowledge of inputs is one of the most common cases for our scheme when $\theta=0.5$. Regarding the general case of the $\theta$ trigger chance and the derivation of average costs, we leave all the results in Appendix \ref{appendix:costs}. For future reference, we present a useful claim derived from Proposition \ref{prop:intra_costs}.

\begin{proposition}[Probability of Triggering]
\label{prop:trigger_prob}
In the case of serial and honest execution, if the probability of triggering an individual skip gate is $\theta$, then the $k$-th skip gate has a $\theta(1-\theta)^{k-1}$ probability of being triggered during the entire procedure, whereas the whole skipping scheme has a $1-(1-\theta)^{n_c}$ probability to be triggered.
\end{proposition}

A summary of the cost estimation can be found in Table \ref{tab:time_costs}. It is noteworthy that the evaluation costs of skipping scheme decreases quasi-exponentially as $n_c$ grows. 
However, as $n_c$ reflects the number of skip gates in a specific subcircuit, it is highly unpredictable and intractable, depending on the intrinsic characteristics of the circuit. The capacity for skip gates is likely to be greater in circuits with fewer inputs, but this is only empirically proven.
This section focuses on intra-subcircuit costs, but Appendix \ref{appendix:costs} examines inter-subcircuit costs in details.

\begin{table}
\centering
\caption{Estimation of the efficiency of the serial skipping scheme when $\theta=0.5$. 
% An individual gate's evaluation cost is the time cost associated with its operation, which may have two alternatives depending on whether the gate is triggered or untriggered. We ignore the average circuit evaluation cost of the PSS since it shares the same lower bound as the CSS.
}
\label{tab:time_costs}
\begin{tabular}{cccc}
\toprule
\textbf{Costs} & \textbf{Baseline} & \textbf{PSS} & \textbf{CSS} \\
\cmidrule(lr{0.2em}){1-1}
\cmidrule(lr{0.2em}){2-2}
\cmidrule(lr{0.2em}){3-3}
\cmidrule(lr{0.2em}){4-4}
\multicolumn{4}{c}{\textit{per gate}} \\
Gate Space & $2\lambda$ & $4\lambda+2\alpha$ & $5\lambda+3\alpha$ \\
Gate Eval. Time & $t_d$ & $t_d+t_h$ & $t_d+t_h$ \textit{or} $t_d+5t_h$ \\
\midrule
\multicolumn{4}{c}{\textit{per circuit}} \\
Circuit Eval. Time (Worst) & $N t_d$ & $N(t_d + t_h)$ & $N(t_d + 5t_h)$ \\
Circuit Eval. Time (Avg.) & $Nt_d$ & $\le$ CSS & $N t_d\cdot 2^{-n_c} + t_d + 10t_h - \tilde{\bigtheta}(2^{-n_c})$ \\
\bottomrule
\end{tabular}
\end{table}

% \subsubsection{Parallelism for Subcircuits.}
% In a multiprocessor computer, multiple paths of the garbled circuit may be evaluated simultaneously, and the path that activates the skip gate may calculate the result prior to receiving the other inputs for the subcircuits. This significantly reduces overall computational overhead and leads to a faster evaluation of the results.

% Assume the gate-hiding circuit is covered by the skipping scheme with portion $p$, and the general trigger rate is $q$. Let $c$ denote the degree of parallelism. The theoretical upper bound of speedup of the execution based on Amdahl's Law \cite{amdahl1967validity} is $\frac{c}{(1-p)+p(1-q)}=\frac{c}{1-pq}$, compared with the speedup for normal circuits as $c$.

% \subsubsection{Costs of Skipping Scheme.}
% A skipping scheme's costs can be classified as intra- and inter-subcircuit costs. Costs associated with intra-subcircuits are related to the execution time of the subcircuit, whereas costs associated with inter-subcircuits are related to the overall improvement of the circuit.

% Results indicate a significant reduction in intra-subcircuit costs. Furthermore, inter-subcircuit costs are bounded by the sum of all intra-subcircuit costs multiplied by decay weights related to circuit depth. Theorem \ref{thm:upper_bound_inter_costs} implies a general \emph{quasi-exponential} improvement on an inter-subcircuit basis. Additionally, deeper circuits benefit more from the skipping scheme.
% Details can be found in Appendix \ref{appendix:costs}.

%
% Security Analysis
%
\section{Security Analysis}
\label{sec:security_analysis}

\subsection{Correctness}

The correctness of our skipping scheme is evident from the fact that it does not impact the regular garbling scheme. In the case of $\mathcal{S}\gets\genss(1^\lambda,1^\alpha,C,\tilde{C})$, it is clear that $\evalss(1^\lambda,1^\alpha,\tilde{C},\mathcal{S},X)$ yields the same result as $\tilde{C}(X)$. Since skipping is disabled if it is not triggered, we only verify the output when it is activated. Let $R$ be the constant final result.

\paragraph{Correctness of PSS.}
Let $S^\text{PSS}:=(\bsl{t_0},\bsl{t_1},G_0,G_1)$ be the context of the skip gate that passes the trigger test. We assume that the $i$-th ($i\in\{0,1\}$) value pair $v_i$ is the input pair to $\evalpss$ and the ciphertexts $G$ in $\genpss$ are rewritten as follows.
\begin{align*}
t_i\gets \hash(s,v_i),&~b_i\gets \lsb{t_i} \\
G_{\lsb{t_0}}\gets \enc_{v_0}^\lambda(R),&~G_{\lsb{t_1}}\gets \enc_{v_1}^\lambda(R)
\end{align*}
We then receive the result.
\begin{align*}
G_{b_i}&=G_{\lsb{t_i}}\gets \enc_{v_i}^\lambda(R) \\
R&\gets \dec_{v_i}^\lambda(G_{b_i})=\dec_{v_i}^\lambda(\enc_{v_i}^\lambda(R))
\end{align*}

\paragraph{Correctness of CSS.}
Suppose CSS gets triggered until $k=n$. For any $k<n$, the CSS gates should get either $v_3^k$ or $v_4^k$ as input. We first prove that the CSS fails the trigger test when $k=m<n$. Assume that we already have the context values $c_\alpha^{m-1}$ and $c_\lambda^{m-1}$ and the context $S_m^{\text{CSS}}$, the $m$-th trigger test then
\begin{enumerate}[topsep=0pt,itemsep=-1ex,partopsep=1ex,parsep=1ex]
    \item calculates $t^m\gets \hash_\alpha(s^m, c_\alpha^{m-1},v_{3}^m)\text{ or }\hash_\alpha(s^m, c_\alpha^{m-1},v_{4}^m)$,
    \item compares $\bsl{t^m}$ with $\bsl{t_0^m}$ and $\bsl{t_1^m}$,
    \item fails and updates the context values $c_\alpha^m$ and $c_\lambda^m$.
\end{enumerate}
Apparently, the iterative trigger test always fails for $k\in[n-1]$. When $k=n$, the CSS gates receive input pair $v_0^k$ or $v_1^k$. Similar to PSS, for the $n$-th gate that survive the trigger test, let $S^\text{CSS}:=(\bsl{t_0^n},\bsl{t_1^n},t_c^n,G_0^n,G_1^n,G_c^n)$ be the context, $v_i^n$ ($i\in\{0,1\}$) be the received input pair and $c_\alpha^{n-1},c_\lambda^{n-1}$ be the context values. Rewrite $\evalcss$ as follows:
\begin{align*}
t_i^n\gets \hash_\alpha(s^n,c_\alpha^{n-1},v_i^n),&~b_i^n\gets \lsb{t_i^n} \\
G_{\lsb{t_0^n}}^n\gets \enc_{v_0^n}^\lambda(R)\oplus c_\lambda^{n-1},&~G_{\lsb{t_1^n}}^n\gets \enc_{v_1^n}^\lambda(R)\oplus c_\lambda^{n-1}
\end{align*}
Similarly, we obtain the same final output as $\tilde{C}(X)$:
\begin{align*}
G_{b_i^n}^n&=G_{\lsb{t_i^n}}^n\gets \enc_{v_i^n}^\lambda(R) \\
R&\gets \dec_{v_i^n}^\lambda(G_{b_i^n})=\dec_{v_i^n}^\lambda(\enc_{v_i^n}^\lambda(R))
\end{align*}

\subsection{Simulation-Based Security}

Security of the skipping scheme in semi-honest scenarios is divided into two categories: untriggered and triggered. In this proof, we demonstrate that the distribution of \emph{view of the evaluation} of the original garbled circuit is indistinguishable from that with the skipping scheme, following the instructions of Lindell \cite{lindell2009proof,lindell2017simulate}. For simplicity, we only prove the security of PSS. Distribution of the view of algorithm $\game$ is denoted by $\view{\game}$. Also, $\view{\game_1}\equiv\view{\game_2}$ indicates that the views of $\game_1$ and $\game_2$ have (computationally) identical distributions.
\begin{definition}[Security of Evaluation with Auxiliary Scheme]
\label{def:security_auxiliary_scheme}
For a garbling scheme $\garblescheme=(\garble,\encode,\eval,\decode)$, the garbled circuit $\tilde{C}\gets\garble(1^\lambda,C)$, garbled input $X\in\{0,1\}^n$, auxiliary scheme $\game$, simulator $\Sim$ and a random oracle $H$, the advantage of the adversary $\mathcal{A}$ is defined as $\adv_{\ad,\eval,\game,\Sim,\Phi}(\lambda):=$
\begin{align*}
\Big\lvert&\Pr[\ad^H(\view{\eval(\Phi(\tilde{C}),X)},\view{\game(\Phi(\tilde{C}), X)})=1] \\
- &\Pr[\ad^H(\view{\eval(\Phi(\tilde{C}),X)},\Sim(\view{\eval(\Phi(\tilde{C}),X)}))=1]\Big\rvert
\end{align*}
A garbling scheme $\garblescheme$ with an auxiliary scheme $\game$ is secure, if 
there exists a PPT simulator $\Sim$, such that for any garbled circuit $\tilde{C}$, garbled input $X$ and PPT adversary $\ad$, the advantage $\adv_{\ad,\eval,\game,\Sim,\Phi}(\lambda)$ is negligible.
\end{definition}

% \begin{remark}
% As a definition of security, we use a compound view of the normal garbling scheme and the skipping scheme, since the skipping scheme is optional. A malicious adversary may intentionally activate both schemes. Therefore, the security where both schemes are active is much more vital than in normal circumstances.
% \end{remark}

The \emph{view of the skipping scheme} consists of the garbled input $X$, the garbled output $\tilde{C}(X)$, the context of the generated skip gates $\mathcal{S}^{\text{SS}}$, the trigger state $\tau_{\text{st}}$ and the indices of the trigger gates $\tau_{\text{idx}}$. In the \emph{view of the garbled circuit}, these variables have been identified: the garbled input and output $X,\tilde{C}(X)$, the information of circuit $\Phi(\tilde{C})$, and the intermediate result $I(\tilde{C},X)$. The security parameter $\lambda$ and $\alpha$ have been omitted for simplicity.
\begin{align*}
\view{\game(\Phi(\tilde{C}), X)}&:=(X,\tilde{C}(X),\mathcal{S}^{\text{SS}},\tau_{\text{st}},\tau_{\text{idx}}) \\
\view{\eval(\Phi(\tilde{C}),X)}&:=(X,\tilde{C}(X),\Phi(\tilde{C}),I(\tilde{C},X))
\end{align*}

Simulating the view of the scheme is achieved with the following oracles:
\begin{itemize}
    \item $\oracle_\tau(\Phi(\tilde{C}))$: Oracle $\oracle_\tau$ is queried to determine whether the skipping scheme has been triggered.
    \item $\oracle_0(\Phi(\tilde{C}))$: Oracle $\oracle_0$ is queried to retrieve the selected wire pairs $\mathcal{P}$.
    \item $\oracle_1(\Phi(\tilde{C}),i,j)$: Oracle $\oracle_1$ is queried to obtain trigger values $t$ and salts $s$ that satisfy the partial inequality criteria.
    \item $\oracle_2(\Phi(\tilde{C}))$: Oracle $\oracle_2$ is queried to obtain the indices of triggered gates $\tau_\text{idx}$.
\end{itemize}
% To simplify the proof of security, we introduce the Oracle $\oracle_\tau$ to obtain the actual trigger state $\tau$, and split the proof of security into two categories: untriggered cases and triggered cases. We defer to section \ref{sec:oracle_to_trigger_state}, the proof of security of introducing the $\oracle_\tau$.

\subsubsection{Untriggered Cases.}
% To demonstrate that the views of a skipping scheme can be simulated only with the view of the execution of the garbled circuit, we simulate both the generation and evaluation of skipping schemes. 

\begin{theorem}[Simulation-based Privacy of Untriggered Cases]
\label{thm:untriggered_privacy}
Let $n_c$ be the number of skip gates.
The general skipping scheme $\game$ described in Section \ref{sec:skipping_scheme} is simulation-based secure in \emph{untriggered} cases, for $\Phi=\topo=(n,m,l,A,B,\mathcal{W})$, where $\mathcal{W}$ represents the mapping function for ciphertexts of each gate. More precisely, there exists an adversary $\ad$ such that
\begin{align*}
\adv_{\ad,\eval,\game,\Sim,\topo}^\text{UT}(\lambda)\le n_c 2^{-(\lambda-1)}.
\end{align*}
\end{theorem}
% Proof of "Simulation-based Privacy of Untriggered Cases"
\begin{proof}
Simulated experiments for untriggered cases are conducted using the simulator $S_{\text{UT}}$, as well as the games $\game_0^{\oracle_0,\oracle_1}$ and $\game_{\text{UT}}$, and described as follows.
\begin{itemize}
    \item $S_{\text{UT}}(\view{\eval(\Phi(\tilde{C}), X)})$: $S_{\text{UT}}$ described in Fig. \ref{fig:ut_game_simulator} generates the context of skip gates and simulates the view of the actual skipping scheme.
    \item $\game_0^{\oracle_0,\oracle_1}(\Phi(\tilde{C}), X)$: $\game_0^{\oracle_0,\oracle_1}$ given in Fig. \ref{fig:ut_game_0} simulate the context of a skipping scheme with information on the garbled circuit and the garbled input $X$. In this game, we utilize oracles $\oracle_0$ and $\oracle_1$ to show the robustness of the simulator by fixing part of the generation of skip gates.
    \item $\game_{\text{UT}}(\Phi(\tilde{C}), X)$: Real experiment of the skipping scheme in untriggered cases.
\end{itemize}

\begin{figure}[htb]
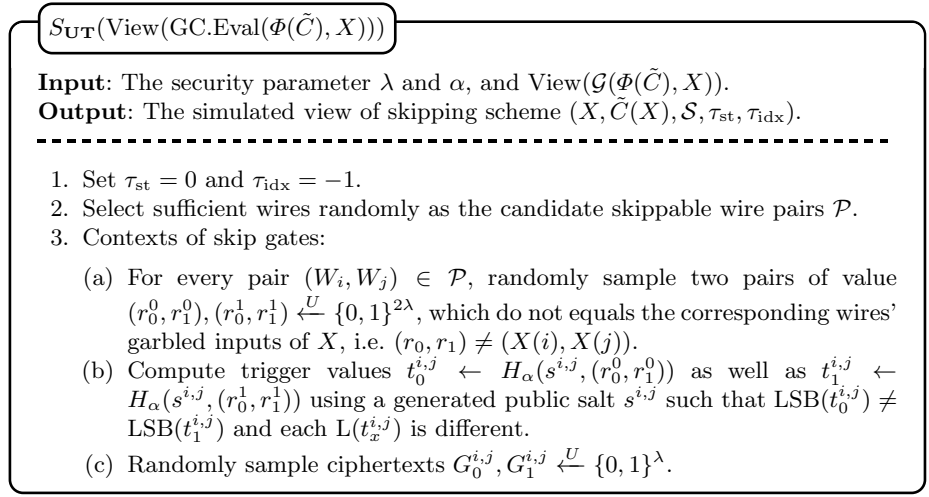

\centering
% Simulator in untriggered case
\begin{schemealg}[$S_{\text{UT}}(\view{\eval(\Phi(\tilde{C}), X)})$]
\textbf{Input}: The security parameter $\lambda$ and $\alpha$, and $\view{\game(\Phi(\tilde{C}), X)}$.

\textbf{Output}: The simulated view of skipping scheme $(X,\tilde{C}(X),\mathcal{S},\tau_{\text{st}},\tau_{\text{idx}})$.
\newline
\hdashrule[0.5ex]{\linewidth}{1pt}{1mm}

\begin{enumerate}[topsep=0pt,itemsep=-1ex,partopsep=1ex,parsep=1ex]
    \item Set $\tau_\text{st}=0$ and $\tau_\text{idx}=-1$.
    \item Select sufficient wires randomly as the candidate skippable wire pairs $\mathcal{P}$.
    \item Contexts of skip gates:
    \begin{enumerate}
        \item For every pair $(W_i,W_j)\in \mathcal{P}$, randomly sample two pairs of value $(r_0^0, r_1^0),(r_0^1, r_1^1)\xleftarrow{U} \{0,1\}^{2\lambda}$, which do not equals the corresponding wires' garbled inputs of $X$, i.e. $(r_0, r_1)\ne (X(i), X(j))$.
        \item Compute trigger values $t_0^{i,j}\gets \hash_\alpha(s^{i,j}, (r_0^0, r_1^0))$ as well as $t_1^{i,j}\gets \hash_\alpha(s^{i,j}, (r_0^1, r_1^1))$ using a generated public salt $s^{i,j}$ such that $\lsb{t_0^{i,j}}\ne \lsb{t_1^{i,j}}$ and each $\bsl{t_x^{i,j}}$ is different.
        \item Randomly sample ciphertexts $G_0^{i,j},G_1^{i,j}\xleftarrow{U} \{0,1\}^\lambda$.
    \end{enumerate}
\end{enumerate}
\end{schemealg}
\caption{The simulator $S_{\text{UT}}$ in untriggered cases.}
\label{fig:ut_game_simulator}
\end{figure}

\begin{figure}[htbp]
\centering
% Game in untrigger case
\begin{schemealg}[$\game_0^{\oracle_0,\oracle_1}(\Phi(\tilde{C}), X)$]
\textbf{Input}: The security parameter $\lambda$ and $\alpha$, $\Phi(\tilde{C})$, and the garbled input $X$.

\textbf{Output}: The simulated view of this game $(X,\tilde{C}(X),\mathcal{S},\tau_{\text{st}},\tau_{\text{idx}})$.
\newline
\hdashrule[0.5ex]{\linewidth}{1pt}{1mm}

\begin{enumerate}[topsep=0pt,itemsep=-1ex,partopsep=1ex,parsep=1ex]
    \item Set $\tau_\text{st}=0$ and $\tau_\text{idx}=-1$.
    \item Retrieve the selected wire pairs from oracle $\mathcal{P}\gets \oracle_0(\Phi(\tilde{C}))$.
    \item Contexts of skip gates:
    \begin{enumerate}
        \item For every pair $(W_i,W_j)\in\mathcal{P}$, compute trigger values and salts from oracle: $(t_0^{i,j},t_1^{i,j})\gets \oracle_1(\Phi(\tilde{C}), i,j)$ with the corresponding salt $s^{i,j}$.
        \item Randomly sample ciphertexts $G_0^{i,j},G_1^{i,j}\xleftarrow{U} \{0,1\}^\lambda$.
    \end{enumerate}
\end{enumerate}
\end{schemealg}
\caption{The game $\game_0^{\oracle_0,\oracle_1}$ in untriggered cases.}
\label{fig:ut_game_0}
\end{figure}

Now we will demonstrate that the simulators are indistinguishable from the real protocol.
\begin{itemize}
    \item $\view{S_\text{UT}}\equiv \view{\game_0^{\oracle_0,\oracle_1}}$: The two games differ in two ways. The first is regarding the skip gates chosen, and the second is regarding the generation of the actual trigger values. Since the skipping scheme remains untriggered, any shift in the distribution of the chosen skip gates will not affect the overall view of the skipping scheme. Given that the trigger values are created by the hash function $\hash_\alpha$ uniformly, the distribution between the one derived from oracle $\oracle_1(\Phi(\tilde{C}), i,j)$ and the sample one should be identical.
    \item $\view{\game_0^{\oracle_0,\oracle_1}}\equiv \view{\game_\text{UT}}$: The only difference between these two games is the ciphertexts. Using ciphertexts in real experiments will prevent malicious decryption from obtaining the actual output $\tilde{C}(X)$. However, ciphertexts in simulated games are sampled uniformly and may fail to meet the criteria mentioned above $\dec_v^\lambda(G^{i,j})\ne\tilde{C}(X)$ for all two ciphertexts of the skip gates. The probability of failure is $1-(1-2^{-\lambda})^{2n_c}\le n_c 2^{-(\lambda-1)}$.
\end{itemize}
% In this case, the simulated experiment is indistinguishable from the real experiment except for the negligible probability of $n_c 2^{-(\lambda-1)}$.
\end{proof}

\subsubsection{Triggered Cases.}
% A similar analysis is performed for the triggered cases of the skipping scheme and the untriggered cases.

\begin{theorem}[Simulation-based Privacy of Triggered Cases]
\label{thm:triggered_privacy}
Let $n_c$ be the number of skip gates.
The general skipping scheme $\game$ described in Section \ref{sec:skipping_scheme} is secure based on simulation in the \emph{triggered} cases, for $\Phi=\topo=(n,m,l,A,B,\mathcal{W})$, where $\mathcal{W}$ represents the mapping function for the ciphertexts of each gate. More precisely, there exists an adversary $\ad$ such that
\begin{align*}
\adv_{\ad,\eval,\game,\Sim,\topo}^\text{T}(\lambda)\le (2n_c-1)2^{-\lambda}.
\end{align*}
\end{theorem}
% Proof of "Simulation-based Privacy of Triggered Cases"
\begin{proof}
In triggered cases, we consider simulator $S_\text{T}$ in the simulated experiments of Definition \ref{def:security_auxiliary_scheme}, and the games $\game$ as below.
\begin{itemize}
    \item $S_\text{T}(\view{\eval(\Phi(\tilde{C}), X)})$: $S_{\text{T}}$ in Fig. \ref{fig:t_game_simulator} generates the context of skip gates, the trigger state, and simulates the view of the actual skipping scheme. 
    \item $\game_\tidx{0}^{\oracle_2}(\Phi(\tilde{C}), X)$: $\game_\tidx{0}^{\oracle_2}$ described in Fig. \ref{fig:t_game_0} simulates the context of skip gates. Queries to oracle $\oracle_2$ constrain the selection of triggered gates, while eliminating the shift of distribution.
    \item $\game_\tidx{1}^{\oracle_0,\oracle_1,\oracle_2}(\Phi(\tilde{C}), X)$: $\game_\tidx{1}^{\oracle_0,\oracle_1,\oracle_2}$ in Fig. \ref{fig:t_game_1} simulates the context with two additional oracles. Oracles $\oracle_0$ and $\oracle_1$ in this game simulate the real skipping scheme with partially predetermined contexts, as in the untriggered cases.
    \item $\game_\text{T}$: Real experiment of the skipping scheme in triggered cases.
\end{itemize}

\begin{figure}[htb]
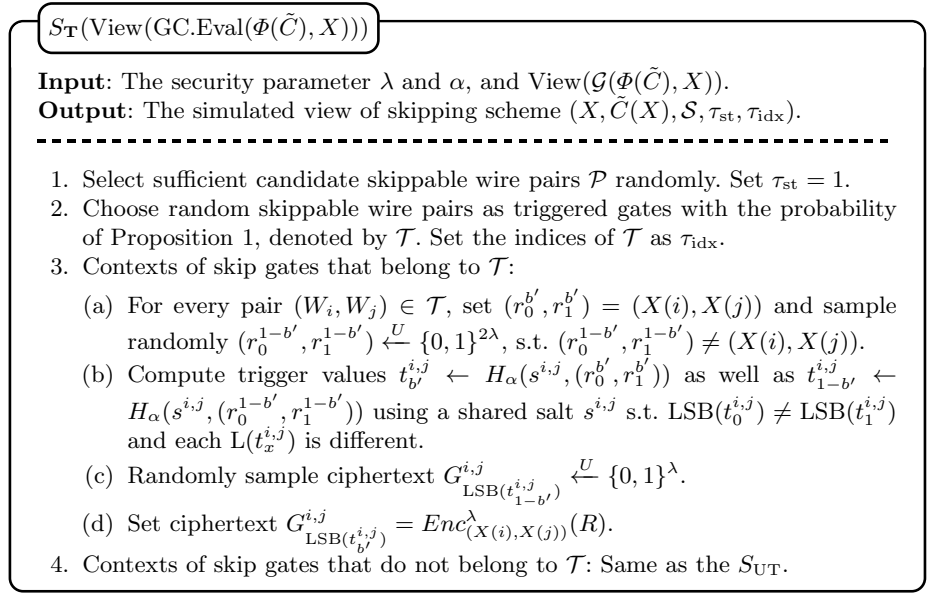

\centering
% Simulator in triggered case
\begin{schemealg}[$S_{\text{T}}(\view{\eval(\Phi(\tilde{C}), X)})$]
\textbf{Input}: The security parameter $\lambda$ and $\alpha$, and $\view{\game(\Phi(\tilde{C}), X)}$.

\textbf{Output}: The simulated view of skipping scheme $(X,\tilde{C}(X),\mathcal{S},\tau_{\text{st}},\tau_{\text{idx}})$.
\newline
\hdashrule[0.5ex]{\linewidth}{1pt}{1mm}

\begin{enumerate}[topsep=0pt,itemsep=-1ex,partopsep=1ex,parsep=1ex]
    \item Select sufficient candidate skippable wire pairs $\mathcal{P}$ randomly. Set $\tau_\text{st}=1$.
    \item Choose random skippable wire pairs as triggered gates with the probability of Proposition \ref{prop:trigger_prob}, denoted by $\mathcal{T}$. Set the indices of $\mathcal{T}$ as $\tau_\text{idx}$.
    \item Contexts of skip gates that belong to $\mathcal{T}$:
    \begin{enumerate}[topsep=0pt,itemsep=-1ex,partopsep=1ex,parsep=1ex]
        \item For every pair $(W_i,W_j)\in \mathcal{T}$, set $(r_0^{b'}, r_1^{b'})=(X(i),X(j))$ and sample randomly $(r_0^{1-b'}, r_1^{1-b'})\xleftarrow{U} \{0,1\}^{2\lambda}$, s.t. $(r_0^{1-b'}, r_1^{1-b'})\ne (X(i),X(j))$.
        \item Compute trigger values $t_{b'}^{i,j}\gets \hash_\alpha(s^{i,j}, (r_0^{b'}, r_1^{b'}))$ as well as $t_{1-b'}^{i,j}\gets \hash_\alpha(s^{i,j}, (r_0^{1-b'}, r_1^{1-b'}))$ using a shared salt $s^{i,j}$ s.t. $\lsb{t_0^{i,j}}\ne \lsb{t_1^{i,j}}$ and each $\bsl{t_x^{i,j}}$ is different.
        \item Randomly sample ciphertext $G_{\lsb{t_{1-b'}^{i,j}}}^{i,j}\xleftarrow{U} \{0,1\}^\lambda$.
        \item Set ciphertext $G_{\lsb{t_{b'}^{i,j}}}^{i,j} = \enc_{(X(i),X(j))}^\lambda(R)$.
    \end{enumerate}
    \item Contexts of skip gates that do not belong to $\mathcal{T}$: Same as the $S_\text{UT}$.
\end{enumerate}
\end{schemealg}
\caption{The simulator $S_{\text{T}}$ in triggered cases.}
\label{fig:t_game_simulator}
\end{figure}

\begin{figure}[htb]
\centering
% Game 0 in triggered case
\begin{schemealg}[$\game_\tidx{0}^{\oracle_2}(\Phi(\tilde{C}), X)$]
\textbf{Input}: The security parameter $\lambda$ and $\alpha$, $\Phi(\tilde{C})$, and the garbled input $X$.

\textbf{Output}: The simulated view of this game $(X,\tilde{C}(X),\mathcal{S},\tau_{\text{st}},\tau_{\text{idx}})$.
\newline
\hdashrule[0.5ex]{\linewidth}{1pt}{1mm}

\begin{enumerate}[topsep=0pt,itemsep=-1ex,partopsep=1ex,parsep=1ex]
    \item Set $\tau_\text{st}=1$ and $\tau_\text{idx}\gets \oracle_2(\Phi(\tilde{C}))$.
    \item Select a sufficient number of wires randomly (with the probability of Proposition \ref{prop:trigger_prob}) containing $\tau_\text{idx}$ as the candidate skippable wire pairs $\mathcal{P}$, i.e. $\tau_\text{idx}\subseteq \mathcal{P}$.
    \item Contexts of skip gates that belong to $\mathcal{T}$: Same as the $S_\text{T}$.
    \item Contexts of skip gates that do not belong to $\mathcal{T}$: Same as the $S_\text{T}$.
\end{enumerate}
\end{schemealg}
\caption{The game $\game_\tidx{0}^{\oracle_2}$ in triggered cases.}
\label{fig:t_game_0}
\end{figure}

\begin{figure}[htbp]
\centering
% Game 1 in triggered case
\begin{schemealg}[$\game_\tidx{1}^{\oracle_0,\oracle_1,\oracle_2}(\Phi(\tilde{C}), X)$]
\textbf{Input}: The security parameter $\lambda$ and $\alpha$, $\Phi(\tilde{C})$, and the garbled input $X$.

\textbf{Output}: The simulated view of this game $(X,\tilde{C}(X),\mathcal{S},\tau_{\text{st}},\tau_{\text{idx}})$.
\newline
\hdashrule[0.5ex]{\linewidth}{1pt}{1mm}

\begin{enumerate}[topsep=0pt,itemsep=-1ex,partopsep=1ex,parsep=1ex]
    \item Set $\tau_\text{st}=1$ and $\tau_\text{idx}\gets \oracle_2(\Phi(\tilde{C}))$.
    \item Retrieve the selected wire pairs from oracle $\mathcal{P}\gets \oracle_0(\Phi(\tilde{C}))$.
    \item Contexts of skip gates that belong to $\mathcal{T}$:
    \begin{itemize}
        \item For every pair $(W_i,W_j)\in\mathcal{P}$, compute trigger values and salts from oracle: $(t_0^{i,j},t_1^{i,j})\gets \oracle_1(\Phi(\tilde{C}), i,j)$ with the corresponding salt $s^{i,j}$. Either $t_0^{i,j}$ or $t_1^{i,j}$ must be equal to the hash of input $t_{b'}^{i,j}:=\hash_\alpha(s^{i,j}, (X(i),X(j)))$.
        \item Randomly sample ciphertext $G_{\lsb{t_{1-b'}^{i,j}}}^{i,j}\xleftarrow{U} \{0,1\}^\lambda$.
        \item Set ciphertext $G_{\lsb{t_{b'}^{i,j}}}^{i,j} = \enc_{(X(i),X(j))}^\lambda(R)$.
    \end{itemize}
    \item Contexts of skip gates that do not belong to $\mathcal{T}$: Same as the $\game_0^{\oracle_0,\oracle_1}$.
\end{enumerate}
\end{schemealg}
\caption{The game $\game_\tidx{1}^{\oracle_0,\oracle_1,\oracle_2}$ in triggered cases.}
\label{fig:t_game_1}
\end{figure}

Next, we will demonstrate that the simulation and the real experiment are indistinguishable. 
\begin{itemize}
    \item $\view{S_\text{T}}\equiv\view{\game_\tidx{0}^{\oracle_2}}$: In the game $\game_\tidx{0}^{\oracle_2}$, oracle $\oracle_2$ is queried to access the indices of the current triggered gates. Since the simulator selects these gates based on the probability of Proposition \ref{prop:trigger_prob}, their views cannot be distinguished from each other.
    \item $\view{\game_\tidx{0}^{\oracle_2}}\equiv\view{\game_\tidx{1}^{\oracle_0,\oracle_1,\oracle_2}}$: As with the untriggered cases, these two games differ in two fundamentally different ways. The distribution of views remains the same.
    \item $\view{\game_\tidx{1}^{\oracle_0,\oracle_1,\oracle_2}}\equiv\view{\game_\text{T}}$: The only difference between these two games is the ciphertexts. The distribution of the ciphertexts is uniformly sampled based on the assumption of a CPA-secure encryption scheme. However, the ciphertexts for invalid inputs in simulated games are sampled uniformly and may fail to reject malicious decryptions. The probability of failure is $1-(1-2^{-\lambda})^{2n_c - |\mathcal{T}|}\le (2n_c - |\mathcal{T}|)2^{-\lambda}\le (2n_c-1)2^{-\lambda}$.
\end{itemize}
% Based on the present results, the simulated experiment is not distinguishable from the real experiment other than the negligible probability of $(2n_c-1)2^{-\lambda}$.
\end{proof}

\subsubsection{Oracle to Trigger State.}
\label{sec:oracle_to_trigger_state}
% The general theorem concerning the security of the skipping scheme with an oracle to the trigger state can be concluded based on the separate discussion described above.

\begin{theorem}[Simulation-based Privacy with Trigger State]
\label{thm:privacy_with_trigger_state}
Let $n_c$ be the number of skip gates.
The general skipping scheme $\game$ described in Section \ref{sec:skipping_scheme} is simulation-based secure with the oracle $\oracle_\tau$ to the trigger state $\tau_\text{st}$, for $\Phi=\topo=(n,m,l,A,B,\mathcal{W})$, where $\mathcal{W}$ represents the mapping function for the ciphertexts of each gate. More precisely, there exists an adversary $\ad$ such that
\begin{align*}
\adv_{\ad,\eval,\game,\Sim,\topo}^{\oracle_\tau}(\lambda)\le (4n_c - 1)2^{-\lambda}.
\end{align*}
\end{theorem}
\begin{proof}
Theorems \ref{thm:untriggered_privacy} and \ref{thm:triggered_privacy} imply the upper bound from $\adv^\text{UT}(\lambda)$ and $\adv^\text{T}(\lambda)$.
\end{proof}

% We will first demonstrate that the garbling scheme with the skipping scheme $\garblescheme^{\skippingscheme}$ is as secure as the classic garbling scheme $\garblescheme$ by reducing the proof of security of $\eval^{\oracle_\tau}(\Phi(\tilde{C}),X)$ to $\eval(\Phi(\tilde{C}),X)$.

\begin{lemma}
\label{lem:gc_oracle_equiv}
The security of evaluation of garbling scheme $\eval^{\oracle_\tau}(\Phi(\tilde{C}),X)$ with the oracle $\oracle_\tau$ to the trigger state is \emph{computationally indistinguishable} from the security of the evaluation of the skipping scheme.
\end{lemma}
% \begin{proof}
% Theorem \ref{thm:privacy_with_trigger_state} states that the view of the skipping scheme can be simulated using only the view of the evaluation of a garbling scheme with oracle $\oracle_\tau$, except for negligible probability. This implies that the security of evaluating a skipping scheme is computationally equivalent to $\eval^{\oracle_\tau}(\Phi(\tilde{C}),X)$.
% \end{proof}

% It follows from Lemma \ref{lem:gc_oracle_equiv} that the security of the skipping scheme can be analyzed in an identical way to assessing the security of the garbling scheme with oracle $\eval^{\oracle_\tau}(\Phi(\tilde{C}),X)$. 
The obliviousness of gate functions $g$ is critical for the security of $\eval^{\oracle_\tau}$. A leaked circuit with gate functions $\Phi_\text{circ}(\tilde{C})$ may allow the evaluator to obtain the triggered gates and acquire the final result with significant probability because of the \emph{circuit satisfiability problems} (CIRCUIT-SAT) or \emph{Boolean satisfiability problems} (BOOLEAN-SAT). Hence, the skipping scheme can only be applied to gate-hiding circuits.

\begin{lemma}[Symmetry of Gate-hiding Garbled Circuits]
\label{lem:symmetry_gc}
Let $C$ be a circuit and $(\tilde{C}, e, d)\gets \garble(1^\lambda, C)$, with $\Phi=\Phi_\text{topo}=(n,m,l,A,B,\mathcal{W})$. Given a fixed garbled input $X$, for any PPT adversary $\ad$ and $y\in \{0,1\}^n$, there exists encoding information $e'$ and negligible $\epsilon(\lambda)$ s.t. $X=\encode(e', y)$ and
\begin{align*}
\Big\lvert&\Pr[(\Phi(\tilde{C}),e,d)\gets\garble(1^\lambda,C):\ad(\Phi(\tilde{C}),e,d)=1] \\ 
-&\Pr[(\Phi(\tilde{C}),e',d)\gets\garble(1^\lambda,C):\ad(\Phi(\tilde{C}),e',d)=1]\Big\rvert\le \epsilon(\lambda)
\end{align*}
\end{lemma}
\begin{proof}
Here we describe only the brief process of proving this using a definite algorithm. 
For the $i$-th single-bit input, denoted by $x(i)$ and $y(i)$ separately, if
\begin{itemize}
    \item $x(i)=y(i)$: remain unchanged for encoding information $e$ for this input wire.
    \item $x(i)\ne y(i)$: swap the garbled values of 0 and 1 in encoding information $e$ for this input wire.
\end{itemize}
It is obvious, then, that the modified encoding information is $e'$ such that $X=\encode(e', y)$. In this instance, the encoding information $e'$ is indistinguishable from the original information.
\end{proof}

% Lemma \ref{lem:symmetry_gc} proves that gate-hiding circuits are more secure than gate-aware circuits. In addition to the secure encoding scheme, the ground-truth inputs are obscured by the additional property of symmetry. In this way, the symmetry of the gate-hiding garbled circuits can be used to demonstrate the security of the oracle $\oracle_\tau$.

\begin{theorem}[Security with Oracle to Trigger State]
\label{thm:security_oracle_tau}
Assuming a gate-hiding garbled circuit $(\Phi(\tilde{C}), e, d)\gets \garble(1^\lambda, C)$ and a random oracle $H$, for any PPT adversary $\ad$ and input $x\in\{0,1\}^n$ and $X\gets\encode(e,x)$, it holds that there exists a negligible $\epsilon(\lambda)$ s.t.
\begin{align*}
\Big\lvert&\Pr[s\gets\decode(d, \eval(\Phi(\tilde{C}),X)):\ad^H(\Phi(\tilde{C}), X, s, d)=x] \\
-&\Pr[s\gets\decode(d, \eval^{\oracle_\tau}(\Phi(\tilde{C}),X)):\ad^{\oracle_\tau}(\Phi(\tilde{C}), X, s, d)=x]\Big\rvert\le \epsilon(\lambda)
\end{align*}
\end{theorem}
\begin{proof}
Lemma \ref{lem:symmetry_gc} states that adversary $\ad^{\oracle_\tau}$ has no advantage over adversary $\ad^H$ in guessing the real inputs since the property of symmetry enables \emph{any} input to hold the same view by hiding the encoding scheme. By doing so, the characteristics of inputs are thoroughly erased, making the trigger state useless. This results in an equal level of security between $\eval^{\oracle_\tau}$ and $\eval$.
\end{proof}

\begin{corollary}[Security of Skipping Scheme]
\label{cor:security_ss}
Except for a negligible probability, the garbling scheme $\garblescheme^{\skippingscheme}$ with the skipping scheme is as secure as the normal garbling scheme $\garblescheme$.
\end{corollary}
% \begin{proof}
% Based on Theorem \ref{thm:privacy_with_trigger_state}, \ref{thm:security_oracle_tau} and Lemma \ref{lem:gc_oracle_equiv}, we demonstrate that garbling schemes with skipping schemes have hybrid security. Thus, both schemes are equally secure.
% \end{proof}

%
% Conclusions
%
\section{Conclusion}
This paper present the skipping scheme for gate-hiding circuits that allows a fast evaluation of garbling schemes on skippable subcircuits. As the number of skip gates increases, the estimated efficiency increases quasi-exponentially, substantially reducing computational overheads. In the presence of oblivious gate functions, our theoretical security analysis proves the hybrid security of our skipping scheme.

% One open question is whether obliviousness is strictly required for gate functions. A key aspect of the proof of the security of a skipping scheme is based on the assumption of the symmetry of gate-hiding circuits. As long as the assumption can be loosened and made compatible with those techniques commonly used in garbled circuits, such as FreeXOR, our approach will be able to achieve greater optimization.

% Another direction is to extend our scheme for gate-aware garbled circuits and generalize the symmetry-based security to avoid BOOLEAN-SAT problems. Techniques of gate-aware garbled circuits reduce space complexity, whereas gate-hiding circuits require tens of times more space. There is still work to be done to optimize gate-hiding cases.

\bibliographystyle{splncs04}
\bibliography{mybibliography}

\appendix
\section{Supplement to Skipping Scheme}
\subsection{Skippable Wires Searching}
\label{appendix:skippable_wires_searching}

This section explains how to select skippable wires in a (sub)circuit $C$. Let $F_C$ be the Boolean function of the circuit, $\mathcal{I}$ be the set of inputs, and $\mathrm{PIS}(\mathcal{I}, F_C)$ be the procedure for searching for prime implicants. In general, when $|\mathcal{I}|$ is small, $\mathrm{PIS}$ can be the Quine-McCluskey algorithm \cite{quine1952problem,mccluskey1956minimization}, while when $|\mathcal{I}|$ is large, $\mathrm{PIS}$ can be the ESPRESSO algorithm \cite{brayton1984logic}.
The overall algorithm for selecting skippable wires is described below.

\begin{figure}[htb]
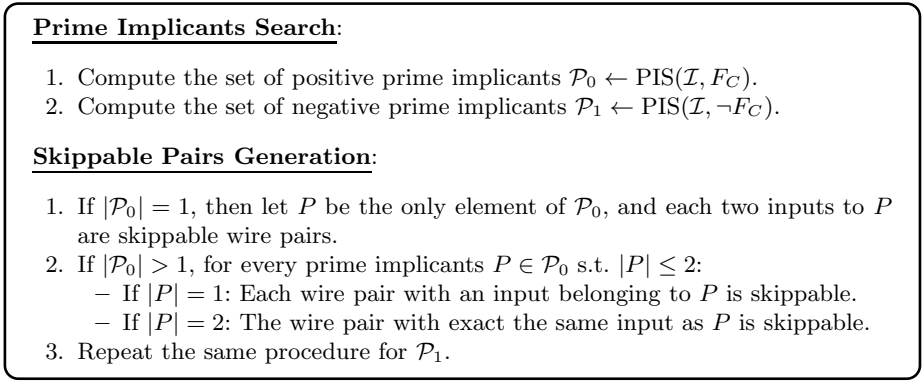

\centering
\begin{normalalg}
\underline{\textbf{Prime Implicants Search}}:
\begin{enumerate}
    \item Compute the set of positive prime implicants $\mathcal{P}_0\gets \mathrm{PIS}(\mathcal{I}, F_C)$.
    \item Compute the set of negative prime implicants $\mathcal{P}_1\gets \mathrm{PIS}(\mathcal{I}, \neg F_C)$.
\end{enumerate}

\underline{\textbf{Skippable Pairs Generation}}:
\begin{enumerate}
    \item If $|\mathcal{P}_0| = 1$, then let $P$ be the only element of $\mathcal{P}_0$, and each two inputs to $P$ are skippable wire pairs.
    \item If $|\mathcal{P}_0| > 1$, for every prime implicants $P\in \mathcal{P}_0$ s.t. $|P|\le 2$:
    \begin{itemize}
        \item If $|P|=1$: Each wire pair with an input belonging to $P$ is skippable.
        \item If $|P|=2$: The wire pair with exact the same input as $P$ is skippable.
    \end{itemize}
    \item Repeat the same procedure for $\mathcal{P}_1$.
\end{enumerate}
\end{normalalg}
\caption{Algorithm for selecting skippable wires.}
\label{fig:skippable_wires_searching}
\end{figure}

\paragraph{Correctness.}
Since each $P$ is an implicant of $F_C$, the corresponding inputs to $P$ would satisfy the skippable wire definition. Thus, each $P$ in the above procedure can be mapped to some candidate skippable wire pairs.

\begin{remark}
As $|P|=2$ indicates, only one pair of inputs can be skipped. It is possible to continue on these wire pairs, but it would reduce the likelihood that skip gates would be triggered. Practically, it would be better to prioritize the wire pairs with $|P|=1$ higher on the priority list.
\end{remark}

\subsection{Costs of Skipping Scheme}
\label{appendix:costs}

Skipping scheme costs can be divided into two categories: intra-subcircuit costs and inter-subcircuit costs. The cost of evaluating a subcircuit falls under the intra-subcircuit cost category, while the inter-subcircuit cost category relates to the improvement of the interplay between subcircuits.

\paragraph{Intra-subcircuit costs.} Our formula for triggering when the chance of triggering is $\theta$ is given here in its complete mathematical derivation. Let $\mathrm{ACC}_{N,n_c}(\theta)$ be the average costs of $N$-gates garbled (sub)circuits with $n_c$ CSS gates.

\begin{proposition}[Intra-subcircuit costs]
\label{prop:intra_costs}
\begin{align*}
\mathrm{ACC}_{N,n_c}(\theta)=&\sum_{k=1}^{n_c}\Pr[k\text{-th gate hit only}]\cdot(k\cdot 5t_h + t_d) + \Pr[\text{All miss}]\cdot N t_d \\
% =& \sum_{k=1}^{n_c}\theta(1-\theta)^{k-1}\cdot(k\cdot t_h + t_d) + (1-\theta)^{n_c}\cdot N\cdot t_d \\
% =& t_d + \frac{t_h}{\theta} + (1-\theta)^{n_c}\cdot(N t_d - t_d - t_h n_c + \frac{t_h}{\theta}) \\
=& t_d + \frac{5t_h}{\theta} + (1-\theta)^{n_c}\cdot N t_d - (1-\theta)^{n_c}\cdot[t_d + 5t_h(n_c+\frac{1}{\theta})]
\end{align*}
\end{proposition}

\paragraph{Inter-subcircuit costs.} Inter-subcircuit costs assume that the entire circuit is formulated as the combinations of indivisible subcircuits. 
A subcircuit's cost can be calculated based on the intra-subcircuit costs mentioned above. We are only interested in the general improvement over the circuit's topology as a result of skipping.

Let $C_i$ denotes the $i$-th subcircuits of the total $n$ subcircutis, and $P(C_i)$ be the set of subcircuits that have outputs to the $C_i$. Let $\mathrm{ACC}(C_i)$ be the intra-subcircuit cost of $C_i$, and $\mathrm{EAC}(C_i)$ be the expected accumulative time cost of $C_i$. Consider all subcircuits as having a probability of $\theta$ of being skipped.
\begin{lemma}[Decay of accumulative costs]
\begin{align*}
\mathrm{EAC}(C_i)=\max_{\hat{C}\in P(C_i)}\{(1-\bigtheta(\theta))\cdot \mathrm{EAC}(\hat{C})\} + \mathrm{ACC}(C_i)
\end{align*}
\end{lemma}
\begin{proof}
The expected accumulative time cost of $C_i$ is determined by the most time-comsuming part of the previous subcircuits. Almost $\bigtheta(\theta)$ portion of the time costs are eliminated as expectation. Considering the decay of time costs from previous subcircuits and the intra-subcircuit costs, the lemma can be formulated as follows.
\end{proof}

The enhancement of inter-subcircuit costs can be measured using $\mathrm{EAC}(C_\text{final})$ in comparison to $\sum_i \mathrm{ACC}(C_i)$, where $C_\text{final}$ is the last subcircuit of the circuit.

\begin{theorem}[Upper bound of $\mathrm{EAC}(C_\text{final})$]
\label{thm:upper_bound_inter_costs}
Let $D(C_i)$ denotes the reverse depth of subcircuit $C_i$ counting from the bottom of the topology of the circuit to the top. For example, $D(C_\text{final})=0$ and $D(C_0)=D_\text{max}$. The upper bound of $\mathrm{EAC}(C_\text{final})$ is described as follows:
\begin{align*}
\mathrm{EAC}(C_\text{final})\le \sum_i B_i(k) \cdot \mathrm{ACC}(C_i)
\end{align*}
where $B_i(k)=\bigtheta((1-\theta)^k)$ is the accumulative decay factor for $C_i$.
\end{theorem}
\begin{proof}
Iteratively expand the $\mathrm{EAC}(C_\text{final})$ from the last subcircuit to the first. The reverse depth determines the decay factor for the intra-subcircuit cost. And the $\mathrm{EAC}(C_\text{final})$ is bounded by the weighted sum over the intra costs.
\end{proof}
\begin{corollary}
In general, deeper circuits benefit more from skipping schemes.
\end{corollary}

\end{document}